\documentclass[12pt,a4,oneside]{article}
\usepackage{amsthm,amssymb,amsmath,amsfonts,float}

\usepackage{mathptmx}

\DeclareSymbolFont{SY}{U}{psy}{m}{n}
\DeclareMathSymbol{\emptyset}{\mathord}{SY}{'306}

\newtheorem{theorem}{Theorem}
\newtheorem{corollary}[theorem]{Corollary}
\newtheorem{lemma}[theorem]{Lemma}

\theoremstyle{definition}

\theoremstyle{remark}
\newtheorem{remark}[theorem]{Remark}
\newtheorem{example}[theorem]{Example}

\begin{document}

\medskip

\begin{center}
{\Large\textsf{ QUANTUM SPEED LIMITS \\[1.5mm]
FOR TIME EVOLUTION OF A SYSTEM SUBSPACE}}
\bigskip

{\large

\textit{Sergio Albeverio}$^\textrm{\,1,*}$, \textit{Alexander K.
Motovilov}$^\textrm{\,2,3,**}$

}

\bigskip

{\small $^\textrm{\,1}$\textsf{Institut f\"ur Angewandte Mathematik
and
HCM, Universit\"at Bonn,\\[-0mm]
Endenicher Allee 60, 53115 Bonn, Germany}
\smallskip

$^\textrm{\,2}$\textsf{Bogoliubov Laboratory of Theoretical Physics,
JINR, Joliot-Curie 6, 141980 Dubna, Russia}
\smallskip

$^\textrm{\,3}$\textsf{Dubna State University, Universitetskaya 19,
141980 Dubna, Russia}

}
\end{center}
\bigskip

\medskip

{\small

By a quantum speed limit one usually understands an estimate on how
fast a quantum system can evolve between two distinguishable states.
The most known quantum speed limit is given in the form of the
celebrated Mandelstam-Tamm inequality that bounds the speed of the
evolution of a state in terms of its energy dispersion. In contrast
to the basic Mandelstam-Tamm inequality, we are concerned not with a
single state but with a (possibly infinite-dimensional) subspace
which is subject to the Schr\"odinger evolution. By using the
concept of maximal angle between subspaces we derive optimal bounds
on the speed of such a subspace evolution. These bounds may be
viewed as further generalizations of the Mandelstam-Tamm inequality.
Our study includes the case of unbounded Hamiltonians.

}

\bigskip

\bigskip

\medskip


\medskip


\bigskip

\vfill

\noindent\rule{35mm}{.2mm}

$^*$E-mail: albeverio@uni-bonn.de

$^{**}$E-mail: motovilv@theor.jinr.ru

\thispagestyle{empty}

\newpage

\section{Introduction}
\label{Intro}

By a \textit{quantum speed limit} one understands a lower bound on
the time required for a quantum system to evolve between two (given)
distinguishable states.  Literature on quantum speed limits is
enormous and in this short article we make no attempt to present a
more or less complete survey of all the results on the subject.
Instead, we simply refer the reader to the recent comprehensive
review articles \cite{DeCa} and \cite{Frey}.

To introduce the main known quantum speed limits, we consider an
isolated quantum system with a Hamiltonian $H$, which is supposed to
be a time-independent self-adjoint operator on the complex Hilbert
space $\mathfrak{H}$. The unit (i.e. norm one) vectors of
$\mathfrak{H}$ represent possible pure states of the system. To be
more precise, a pure state $\mathcal S$ of a quantum system is a
class of equivalence of vectors on a unit sphere in $\mathfrak{H}$:
the unit vectors $\phi,\psi\in\mathfrak{H}$ represent the same state if
there is $\alpha\in[0,2\pi)$ such that
$\psi=\mathrm{e}^{\mathrm{i}\alpha}\phi$. In an obvious way, one may
identify the state $\mathcal S$ with a one-dimensional subspace
$\mathfrak{P}_{\mathcal S}$ which is a linear span of an arbitrarily
chosen vector $\psi$ in $\mathcal S$,\,\, $\mathfrak{P}_{\mathcal
S}:=\bigl\{\phi=\lambda \psi\,\,\, \big|
\lambda\in\mathbb{C}\,\bigr\}$.

In the following, we will assume, for shortness, that the
measurement units are chosen in such a way that $\hbar=1$. It is
supposed that the evolution of a state vector
$\psi(t)\in\mathfrak{H}$, $t\in\mathbb{R}$, is governed by the
Schr\"odinger equation
\begin{align}
\label{Sch1}
\mathrm{i} \frac{d}{dt}\psi & =H\psi,\\
\label{Sch2} \psi(t)\bigg|_{t=t_0}&=\psi_0,
\end{align}
where the vector $\psi_0$ is taken in the domain $\mathrm{Dom}(H)$
of $H$. This vector represents an initial state of the system.

Let $t_0=0$. Then the solution to \eqref{Sch1}, \eqref{Sch2} is
given by
\begin{equation}
\label{U} \psi(t)=U(t)\psi_0,
\end{equation}
where
\begin{equation}
 U(t)=\mathrm{e}^{-\mathrm{i} Ht}, \quad t\in\mathbb{R},
\end{equation}
form a strongly continuous unitary group. Studies of quantum speed
limits originate from the very basic question: \textit{How fast can
the isolated system with the Hamiltonian $H$ evolve to a state
orthogonal to its initial state $\psi_0$? }

The importance of this question is obvious in many respects.
Probably, the most recent motivation comes from quantum information
theory and quantum computing (see, e.g., \mbox{\cite{DeCa,Frey}}).

Known answers to the above basic question have been given in the
form of lower bounds for the so-called orthogonalization time
$T_\perp$, that is, for the time needed for the system to evolve
from $\psi_0$ to a state $\psi(T_\perp)$ such that ${\langle}
\psi_0,\psi(T_\perp){\rangle}=0$.

The first of these bounds is the celebrated \textit{Mandelstam--Tamm
inequality} of 1945 discovered in \cite{MaT}:
\begin{equation}
\label{MT} T_\perp\geq \frac{\pi}{2\, \Delta E},
\end{equation}
where $\Delta E$ is the energy spread (dispersion) in the initial
state $\psi_0$,
\begin{equation}
\Delta E=\sqrt{\| H\psi_0\|^2 -{\langle} H\psi_0,\psi_0{\rangle}^2},
\quad \psi_0\in\mathrm{Dom}(H).
\end{equation}

Another lower bound for the orthogonalization time, the
\textit{Margolus--Levitin inequality} \cite{MaLe} has been found
more than half a century later, in 1998. This bound has the
following form:
\begin{equation}
\label{ML} T_\perp\geq \frac{\pi}{2\, \delta E},
\end{equation}
where
\begin{equation}
\delta E= {\langle}
H\psi_0,\psi_0{\rangle}-\min\bigl(\mathop{\mathrm{spec}}(H)\bigr)
\end{equation}
is nothing but the average energy for the state $\psi_0$ measured
relative to the lower edge of the spectrum
$\mathop{\mathrm{spec}}(H)$ of the Hamiltonian $H$.

Both the bounds \eqref{MT} and \eqref{ML} have been proven to be
sharp (see, e.g., \cite[p. 7]{DeCa} and \cite[p. 3923]{Frey}).

It is worth to remark that the inequalities \eqref{MT} and
\eqref{ML} recall the uncertainty relation for energy and time but
are very different from this relation in the essence since both
\eqref{MT} and \eqref{ML} are related not to the standard deviation
in the measurement of $t$ but to the well-founded time for a given
state to evolve into an orthogonal state.

There is also a lower bound for intermediate time moments, namely,
the \textit{Fleming bound} (1973) derived in \cite{Flem}:
\begin{equation}
\label{Flem} T_\theta\geq \frac{\theta}{\Delta E}\,,
\end{equation}
where $T_\theta$ stands for the time moment at which the acute angle
$$\angle\bigl(\psi_0,\psi(t)\bigr):=\arccos|{\langle}\psi_0,\psi(t){\rangle}|$$
between the states represented by the vectors $\psi_0$ and $\psi(t)$
reaches a certain value $\theta\in(0,\pi/2]$.

Obviously, the Mandelstam-Tamm bound represents a particular case of
the Fleming bound for $\theta=\frac{\pi}{2}$. In fact, through the
years, the Mandelstam-Tamm inequality \eqref{MT} has been
rediscovered several times by various researchers (for the
corresponding discussion, see, e.g., \cite[p.\,5]{DeCa}). Also,
there are generalizations of this bound to the evolution of mixed
states. Furthermore, there are more detailed evolution speed
estimates for particular classes of quantum-mechanical evolutionary
problems (see \cite{DeCa,Frey}). The existence of
Mandelstam-Tamm-type bounds for the orthogonalization time is
discussed even for certain non-self-adjoint (so-called
pseudo-Hermitian and, in particular, $PT$-symmetric) Hamiltonians
(see \cite{BBr,China} and references therein).

In the present work, we are concerned with a possibly
infinite-dimensional subspace of the system states that evolve
accordingly to the non-stationary Schr\"odinger equation. Notice
that by a subspace of the Hilbert space $\mathfrak{H}$ we always
understand a closed linear subset in $\mathfrak{H}$. By using the
concept of maximal angle between subspaces we then derive several
estimates on the speed of such a subspace evolution. The estimate
\eqref{TtPb1} we attain in Theorem~\ref{TTtheta} below  may be
viewed as a natural extension of the Fleming bound \eqref{Flem}.

\section{Bounds for the speed of the subspace evolution}

We are concerned not with a single state but with a whole (possibly
infinite-dimensional) subspace spanned by the system states that are
subject to the Schr\"odinger evolution. That is, we consider a
subspace $\mathfrak{P}_0\subset\mathfrak{H}$ every element $\psi_0$
of which is subject to the Schr{\"o}dinger evolution \eqref{Sch1},
\eqref{Sch2}, i.e.,
\begin{align}
\label{Sch1p}
\mathrm{i} \frac{d}{dt}\psi & =H\psi,\\
\label{Sch2p} \psi(t)\bigg|_{t=0}&=\psi_0, \quad \psi_0 \in
\mathfrak{P}_0.
\end{align}
For simplicity (in fact, mainly for avoiding discussion of the
domains), the Hamiltonian $H$ is assumed to be a bounded operator.

Given $t\geq 0$, by $\mathfrak{P}(t)$ we denote the subspace of
$\mathfrak{H}$ spanned by the values $\psi(t)$ of the vector-valued
functions that solve \eqref{Sch1p}, \eqref{Sch2p} for various
$\psi_0\in\mathfrak{P}_0$. So that we deal with a path $
\mathfrak{P}(t),\,\,  t\geq 0, $ in the set of all subspaces of the
Hilbert space $\mathfrak{H}$. Or (and this is the same) with the
path
\begin{equation}
\label{Ppath} P(t), \quad t\geq 0,\qquad
\mathop{\mathrm{Ran}}\bigl(P(t)\bigr)=\mathfrak{P}(t),
\end{equation}
of the orthogonal projections $P(t)$ in $\mathfrak{H}$ onto the
respective subspaces $\mathfrak{P}(t)$.

It is well known, and this is easily verified by inspection, that
the projection path $P(t)$ is the (unique) solution to the Cauchy
problem
\begin{align}
\label{Sch1P}
\mathrm{i} \frac{d}{dt} P & =[P,H],\\
\label{Sch2P} P(t)\bigg|_{t=0}&=P_0,
\end{align}
where $[P,H]:=PH-HP$ \,\, denotes the commutator of $P=P(t)$ and
$H$, and  $ \mathop{\mathrm{Ran}}\big(P_0\bigr)=\mathfrak{P}_0$. The
solution to \eqref{Sch1P}, \eqref{Sch2P} is explicitly given by
\begin{equation}
\label{PT} P(t)=U(t)P_0 U(t)^*=\mathrm{e}^{-\mathrm{i}
Ht}P_0\mathrm{e}^{\mathrm{i} Ht}.
\end{equation}

We further notice that the set of all orthogonal projections in the
Hilbert space $\mathfrak{H}$ (and hence the set of all subspaces of
$\mathfrak{H}$) is a metric space with distance given by the
standard operator norm,
$$
{\rho}(Q_1,Q_2):=\|Q_1-Q_2\|, \qquad
{\rho}(\mathfrak{Q}_1,\mathfrak{Q}_2):={\rho}(Q_1,Q_2),
$$
where $Q_1$, $Q_2$ are arbitrary orthogonal projections and
$\mathfrak{Q}_1$, $\mathfrak{Q}_2$, their ranges.

It is, however, much less known that there is another natural metric
on the set of all orthogonal projections in $\mathfrak{H}$ (and
hence on the set of all the subspaces of $\mathfrak{H}$). The
corresponding distance is defined by
\begin{equation}
\label{rhoM}
\vartheta(\mathfrak{Q}_1,\mathfrak{Q}_2):=\vartheta(Q_1,Q_2):=\arcsin(\|Q_1-Q_2\|).
\end{equation}
That \eqref{rhoM} is a metric has been proven in 1993 by Lawrence
Brown \cite{Brown}. An alternative proof may be found in
\cite{AM-CAOT}.

The quantity {$\vartheta(\mathfrak{Q}_1,\mathfrak{Q}_2)$} is called
the \textit{maximal angle} between the subspaces $\mathfrak{Q}_1$
and $\mathfrak{Q}_2$.
\bigskip

\begin{remark}
The concept of maximal angle between subspaces can be traced back to
Krein, Krasnoselsky, and Milman \cite{KKM1948}. Assuming that
$(\mathfrak{Q}_1,\mathfrak{Q}_2)$ is an ordered pair of subspaces
with $\mathfrak{Q}_1\neq\{0\}$, they applied the notion of the
(relative) maximal angle between $\mathfrak{Q}_1$ and
$\mathfrak{Q}_2$ to the number
$\varphi(\mathfrak{Q}_1,\mathfrak{Q}_2)\in[0,\pi/2]$ such that
\begin{equation}
\label{t12}
\sin\varphi(\mathfrak{Q}_1,\mathfrak{Q}_2)=\sup\limits_{x\in\mathfrak{Q}_1,\,\|x\|=1}\mathop{\rm
dist}(x,\mathfrak{Q}_2).
\end{equation}
If both $\mathfrak{Q}_1\neq\{0\}$ and $\mathfrak{Q}_2\neq\{0\}$ then
\begin{equation}
\label{tet}
\vartheta(\mathfrak{Q}_1,\mathfrak{Q}_2)=\max\bigl\{\varphi(\mathfrak{Q}_1,\mathfrak{Q}_2),
\varphi(\mathfrak{Q}_2,\mathfrak{Q}_1)\bigr\}.
\end{equation}
Unlike $\varphi(\mathfrak{Q}_1,\mathfrak{Q}_2)$, the maximal angle
$\vartheta(\mathfrak{Q}_1,\mathfrak{Q}_2)$ is always symmetric with
respect to the interchange of the arguments $\mathfrak{Q}_1$ and
$\mathfrak{Q}_2$. Furthermore,
$$
\varphi(\mathfrak{Q}_2,\mathfrak{Q}_1)=\varphi(\mathfrak{Q}_1,\mathfrak{Q}_2)=
\vartheta(\mathfrak{Q}_1,\mathfrak{Q}_2)\quad \text{whenever
}\|Q_1-Q_2\|<1.
$$
\end{remark}

To give a quantum-mechanical interpretation of the maximal angle
between subspaces we follow the concept of a subspace-state of a
quantum system. Namely, given a subspace
$\mathfrak{Q}\subset\mathfrak{H}$, one says that the system is in
the $\mathfrak{Q}$-state if it is in a pure state described by a
(non-specified) normalized vector $x\in\mathfrak{Q}$. Clearly, by
\eqref{t12} and \eqref{tet}
 the quantity $\cos^2\theta(\mathfrak{Q}_1,\mathfrak{Q}_2)$ may be understood as a
minimum probability for a quantum system which is in a
$\mathfrak{Q}_1$-state to be found also in a $\mathfrak{Q}_2$-state.

Now assume that $Q(t)$, $t\geq 0$, is an arbitrary piecewise smooth
path in the set of orthogonal projections on the Hilbert space
$\mathfrak{H}$. By using the triangle inequality for subspaces one
verifies (see, e.g., \cite[Theorem 1]{MS2015}) that
\begin{equation}
\label{tPP}
\vartheta\bigl(\mathfrak{Q}_0,\mathfrak{Q}_t\bigr)\leq\int\limits_0^t
\|\dot Q(\tau)\| d\tau,
\end{equation}
where $\dot Q(t)=\frac{dQ(t)}{dt}$ and $\mathfrak{Q}_t:=
\mathop{\mathrm{Ran}}\bigl(Q(t)\bigr)$, $t\geq 0$.

The following statement is the main tool that we use below in
establishing quantum speed limits for evolution of subspaces. It
represents itself the first of such limits.

\begin{theorem}
\label{Th2} Assume that $P(t)$, $t\geq 0$, is the path \eqref{PT}
where $P_0$ is an orthogonal projection in $\mathfrak{H}$. Then the
following inequality holds
\begin{equation}
\label{MBound}
\vartheta\bigl(\mathfrak{P}_0,\mathfrak{P}(t)\bigr)\leq
V_{{H,P_0}}\, t,
\end{equation}
where $\mathfrak{P}_0= \mathop{\mathrm{Ran}}(P_0)$,
$\mathfrak{P}(t)= \mathop{\mathrm{Ran}}\bigl(P(t)\bigr)$, $t\geq 0$,
and
\begin{equation}
\label{VHP} V_{{H,P_0}}:=\|P_0 H P_0^\perp\|=\|P_0^\perp H P_0\|.
\end{equation}
\end{theorem}

\begin{proof}
In the case under consideration the operator norm of the commutator
of $P(t)$ and $H$ does not depend on $t$, thus,
\begin{equation}
\label{not} \bigl\|[P(t),H]\bigr\|=\bigl\|[P_0,H]\bigr\|, \quad
t\geq 0.
\end{equation}
Furthermore, the commutator $[P_0,H]$ is block off-diagonal with
respect to the orthogonal decomposition
$\mathfrak{H}=\mathfrak{P}_0\oplus\mathfrak{P}_0^\perp$, more
precisely,
\begin{equation}
\label{PHHH} [P_0,H]=P_0HP_0^\perp-P_0^\perp HP_0.
\end{equation}
From \eqref{PHHH} it immediately follows that
$\bigl\|[P_0,H]\bigr\|=\|P_0HP_0^\perp\|=\|P_0^\perp HP_0\|$. Thus,
in order to conclude with \eqref{MBound} it only remains to combine
\eqref{Sch1P} with \eqref{not} and \eqref{PHHH} and then to apply to
$P(t)$ the inequality \eqref{tPP}.

The proof is complete.
\end{proof}

It is worth to notice that by \eqref{MBound}, \eqref{VHP} only the
off-diagonal entries  $P_0HP_0^\perp$ and $P_0^\perp HP_0$ of $H$
contribute into the variation of the subspace $\mathfrak{P}_0$. If
the Hamiltonian $H$ is block diagonal with respect to the
decomposition
$\mathfrak{H}=\mathfrak{P}_0\oplus\mathfrak{P}_0^\perp$ and, thus,
the subspace $\mathfrak{P}_0$ is reducing for $H$ it does not vary
in time at all. This concerns, in particular the case where
$\mathfrak{P}_0$ is a spectral subspace of $H$.

\begin{corollary}
\label{Cor1} Under the hypothesis of Theorem \ref{Th2}, assume that
$T_\theta$ is a time moment for which the maximal angle between the
initial subspace $\mathfrak{P}_0$ and a subspace in the path
$\mathfrak{P}(t)$, $t\geq 0$, reaches the value of $\theta$,
$0<\theta\leq \frac{\pi}{2}$, that is,
\begin{equation}
\label{TtP}
\vartheta\bigl(\mathfrak{P}_0,\mathfrak{P}(T_\theta))=\theta.
\end{equation}
Then
\begin{equation}
\label{TtPb} T_\theta\geq \frac{\theta}{V_{{H,P_0}}}\,.
\end{equation}
\end{corollary}

\begin{example}
\label{exmpl} Let the Hamiltonian $H$ correspond to a two-level
system with non-degenerate bound states $e_1$ and $e_2$, that is,
$\|e_1\|=\|e_2\|=1$, ${\langle} e_1,e_2{\rangle}=0$, and
$$
H=E_1 {\langle} \cdot,e_1{\rangle} e_1+ E_2 {\langle}
\cdot,e_2{\rangle} e_2
$$
where the binding energies $E_1$ and $E_2$ are different, $E_1\neq
E_2$. Assume that $P_0$, $P_0={\langle} \cdot,e{\rangle} e$,  is
projection on the one-dimensional subspace spanned by the vector
$e=\frac{1}{\sqrt{2}}(e_1+e_2)$.
\end{example}

Notice that Example \ref{exmpl} is employed in many papers on
quantum speed limits (see, e.g., \cite{DeCa,Frey,BBr,China}). In
particular, this example proves tightness of both the
Mandelstam-Tamm and Margolus-Levitin inequalities (see, e.g.,
\cite[Section 2.4]{DeCa}). One easily verifies that this example
also works well for the bound \eqref{TtPb} turning this bound into
equality. That is, the bound \eqref{TtPb} is optimal.

\begin{theorem}
\label{TTtheta} Assume the hypothesis of Theorem \ref{Th2}. Let
$\theta$ and $T_\theta$ be the same as in Corollary \ref{Cor1}. Then
the following inequality holds:
\begin{equation}
\label{TtPb1} T_\theta\geq \frac{\theta}{\Delta
E_{\mathfrak{P}_0}}\,,
\end{equation}
where
\begin{equation}
\label{Delta} \Delta
E_{\mathfrak{P}_0}:=\sup\limits_{\psi\in\mathfrak{P}_0,
\,\|\psi\|=1} \bigl({\langle} H^2\psi,\psi{\rangle}-{\langle}
H\psi,\psi{\rangle}^2\bigr)^{1/2}
\end{equation}
\end{theorem}

Skipping the proof, we only notice that \eqref{TtPb1} is proven by
Theorem \ref{Th2} by taking into account that $V_{{H,P_0}}\leq
\Delta E_{\mathfrak{P}_0}$.

\begin{remark}
Example \ref{exmpl} shows that the bound \eqref{TtPb1} is sharp. (In
the case of a one-dimensional subspace $\mathfrak{P}_0$ this bound
simply turns into the Fleming bound for the speed of a state
evolution).
\end{remark}

The next statement is easy to prove and it is rather well known. We
present it here only for convenience of the reader.

\begin{lemma}
\label{LEmEm} For the maximal energy dispersion $\Delta
E_{\mathfrak{P}_0}$ on the subspace $\mathfrak{P}_0$ defined by
\eqref{Delta}, one always has the (optimal) bound
\begin{equation}
\label{Ebnd} \Delta E_{\mathfrak{P}_0}\leq \frac{E_{\rm
max}(H)-E_{\rm min}(H)}{2},
\end{equation}
where
$$
E_{\rm min}(H)=\min\bigl(\mathop{\mathrm{spec}}(H)\bigr)\quad\text{
and }\quad E_{\rm max}(H)=\max\bigl(\mathop{\mathrm{spec}}(H)\bigr)
$$
are respectively the upper and lower bounds of the spectrum of the
Hermitian Hamiltonian~$H$.
\end{lemma}

With Lemma \ref{LEmEm} one easily obtains the following corollary to
Theorem \ref{TTtheta}.

\begin{corollary}
\label{CorFin} Assume that $\Omega$ is a non-negative number and let
$\mathcal{B}_{\Omega}(\mathfrak{H})$ be the set of all bounded
self-adjoint operators $H$ in the Hilbert space $\mathfrak{H}$
$($with $\mathop{\rm dim}\mathfrak{H}\geq 2$$)$ such that
$$
E_{\rm max}(H)-E_{\rm min}(H)=\Omega.
$$
Then
\begin{equation}
\label{TtOm} \inf_{H\in\mathcal{B}_\Omega(\mathfrak{H})} T_\theta(H)
= \frac{2\theta}{\Omega},
\end{equation}
where $T_\theta(H)$ is a time moment for which the maximal angle
between the initial subspace $\mathfrak{P}_0$ and a subspace in the
path $\mathfrak{P}(t)$, $t\geq 0$, given by \eqref{PT} reaches the
value of $\theta\leq\frac{\pi}{2}$.
\end{corollary}

The bound \eqref{TtOm} represents a generalization to subspaces of
the optimal passage time estimate established for the quantum
brachistochrone problem (see, e.g. \cite[equation (12.17)]{BBr}).
The latter estimate is nothing but the equality in the Fleming bound
\eqref{Flem} with $\Delta E$ replaced by $\frac{1}{2}\Omega$ where
the quantity $\Omega$ is introduced in Corollary \ref{CorFin}.
\medskip

\noindent\textbf{Acknowledgments.} Support of this work by the
Heisenberg-Landau Program is kindly acknowledged.

%%%%%%%%%%%%%%%%%%%%%%%%%%%%%%%%%%%%%%%%%%%%%%%%%%%%%%%%%%%%%%%%%

\end{document}